\documentclass{llncs}

\usepackage[usenames,dvipsnames]{xcolor}
\usepackage[utf8]{inputenc}
\usepackage{amsmath,amsfonts,amssymb}
\usepackage{stmaryrd,latexsym}
\usepackage{theorem}
\usepackage{cite}
\usepackage{enumerate,hyperref}
\usepackage{verbatim}
\usepackage{graphicx}
\usepackage{tikz}
\usepackage{lscape}
\usetikzlibrary{arrows}
\usepackage{listings}
\usetikzlibrary{calc,fit,shapes,positioning,arrows,external}

\newcommand{\inference}[3][]{\text{#1}\frac{\displaystyle #2}{\displaystyle #3}}

\newcommand{\E}[2]{\{{#1}\}_{#2}}
\newcommand{\D}[2]{D_{#2}({#1})}

\newcommand{\Prt}{\mathcal{P}}

\newcommand{\M}{\mathcal{M}}
\newcommand{\N}{\mathcal{N}}
\newcommand{\Hedges}{\mathcal{H}}
\newcommand{\R}{\mathcal{R}}
\newcommand{\HB}{\mathcal{HB}}
\newcommand{\V}{\mathcal{V}}

\newcommand{\ds}{\operatorname{ds}}
\newcommand{\cc}{\operatorname{mcd}}
\newcommand{\coa}{\operatorname{ad}}
\newcommand{\cri}{\operatorname{CD}}
\newcommand{\mdd}{\operatorname{mdd}}
\newcommand{\pr}{\operatorname{pr}}
\newcommand{\fv}{\operatorname{fv}}
\newcommand{\fn}{\operatorname{fn}}

\newcommand\sd[1]{[\![{#1}]\!]}

\newcommand{\T}{\mbox{\bf true}}
\newcommand{\F}{\mbox{\bf false}}

\newcommand{\defeq}{\triangleq}
\newcommand{\term}{t}
\newcommand{\typeN}{\mathsf{N}}
\newcommand{\typeB}{\mathsf{B}}
\newcommand{\typeC}{\mathsf{C}}

\newcommand{\fakemiddlepipe}[2]{{\kern-\nulldelimiterspace\left.\vphantom{{#1}{#2}}\right|}}

\makeatletter
\newcommand\xLeftrightarrow[2][]{%
	\ext@arrow 9999{\Longleftrightarrowfill@}{#1}{#2}}
\newcommand\Longleftrightarrowfill@{%
	\arrowfill@\Leftarrow=\Rightarrow}
\newcommand{\xRightarrow}[2][]{\ext@arrow 0359\Rightarrowfill@{#1}{#2}}
\makeatother

\title{Deciding Hedged Bisimilarity}%
\author{Alessio Mansutti \and Marino Miculan}
\institute{
\href{http://mads.uniud.it}{Laboratory of Models and Applications of Distributed Systems}\\
		Department of Mathematics, Computer Science and Physics, 
University of Udine \\
\email{alessio.mansutti@gmail.com} \quad
\email{marino.miculan@uniud.it}
}

\begin{document}
\maketitle

\begin{abstract}
The spi-calculus is a formal model for the design and analysis of cryptographic protocols: many security properties, such as authentication and strong confidentiality, can be reduced to the verification of behavioural equivalences between spi processes.
In this paper we provide an algorithm for deciding \emph{hedged bisimilarity} on finite processes, which is equivalent to barbed equivalence (and coarser than framed bisimilarity). 
This algorithm works with any term equivalence satisfying a simple set of conditions, thus encompassing  many different encryption schemata.  
\begin{keywords}
security, cryptographic protocols, spi-calculus, bisimilarity.
\end{keywords}
\end{abstract}

\section{Introduction}

The spi calculus, introduced by Abadi and Gordon in \cite{ag:spi}, is a process calculus designed for the description and formal verification of cryptographic protocols. Many security properties, such as authentication and strong confidentiality, can be reduced to the verification of \emph{may-testing equivalences} between spi processes.
Since may-testing equivalences are difficult to check in practice, other behavioural equivalences have been put forward for this calculus. In \cite{AbadiG98} Abadi and Gordon defined  \emph{framed bisimilarity}, a bisimulation-style equivalence which is a sound approximation of may-testing equivalence. Later, other context-sensitive equivalences have been proposed; in particular, Borgstr\"om and Nestmann defined \emph{hedged bisimilarity} \cite{bn:mscs05}, which is shown to be equivalent to \emph{barbed equivalence} and strictly coarser than framed bisimilarity. We refer to  \cite{bn:mscs05} for a detailed comparison of these equivalences, which we summarize in Figure~\ref{fig:bisims}.

\begin{figure}[t]
	\centering
	\begin{tikzpicture}[iff/.style={implies-implies,double equal sign distance},
	implies/.style={-implies,double equal sign distance}]
	\node (1) {Framed};
	\node[draw, rectangle] (2) [below right=20pt of 1]{Hedged};
	\node (3) [right=35pt of 2]{Alley};
	\node (4) [right=35pt of 3]{Barbed};
	\node (5) [above right=20pt of 1]{Fenced};
	\node (6) [right=35pt of 5]{Trellis};
	
	\path (1) edge[implies] node {} (2);
	\path (2) edge[iff] node {} (3);
	\path (5) edge[implies] node {} (2);
	\path (5) edge[iff] node {} (6);
	\path (6) edge[implies] node {} (3);
	\path (3) edge[iff] node {} (4);
	\path (5) edge[implies] node {} (1);
	\end{tikzpicture}
	\caption{Behavioural equivalences for the spi calculus \cite{bn:mscs05}.}
	\label{fig:bisims}
\end{figure}
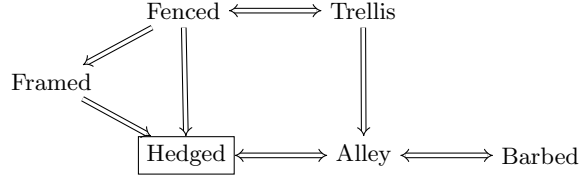

H\"uttel  \cite{huttel:infinity02} proved that framed bisimilarity is decidable on finite processes.
In this paper, we extend this result, proving that also hedged bisimilarity (and hence barbed equivalence) is decidable on finite processes. 
Moreover, we do not choose a specific congruence over terms; rather, the algorithm works with any congruence relation, as long as some mild conditions are satisfied. Therefore, our algorithm can be readily applied to different encryption/decryption schemata just by changing the congruence rules.
These conditions are introduced in Section~\ref{sec:spi}, where we recall also the syntax and \emph{late} operational semantics of spi-calculus.
In Section~\ref{sec:hedged} we define the notion of \emph{hedged bisimilarity} using this late semantics, and in Section~\ref{sec:decid} we show that it is decidable on \emph{finite} spi-calculus processes (i.e. processes without replication); an algorithm in pseudo-code is provided.
Some concluding remarks and directions for future work are in Section~\ref{sec:concl}.

\section{The spi calculus}\label{sec:spi}
The \emph{spi-calculus} extends the $\pi$-calculus with terms and primitives for encryption and decryption.  In this section we define the variant we consider in this paper.

\subsection{Syntax}

\paragraph{Terms} 
We first define the set of terms that can be used by processes, following \cite{bn:mscs05}.

\begin{definition}[Terms]
	Let $\N$ be a countable set of names ranged over by $a,b,c,n\dots$, and $V$ a countable set of variable symbols, ranged over by $x,y,z\dots$.	
	The set of spi-calculus terms is given by the grammar
	\begin{align*}
		A &::= a \mid x \\
		\term &::= A \mid (\term_1, \term_2) \mid \pi_1(\term) \mid \pi_2(\term) \mid \E{\term_1}{\term_2} \mid \D{\term_1}{\term_2}\\
		\phi &::= \T \mid \lnot\phi \mid \phi_1\land\phi_2 \mid [\term_1 = \term_2]
	\end{align*}
\end{definition}
Intuitively, $\E{\term_1}{\term_2}$ denotes the term $\term_1$ encrypted using key $\term_2$, and $(\term_1, \term_2)$ denotes the pair whose components are terms $\term_1$ and $\term_2$. Correspondingly, we have the destructor $\D{\term_1}{\term_2}$, which decrypts $\term_1$ using key $\term_2$, and the two projections $\pi_1(\term)$,  $\pi_2(\term)$. We use $\F$ as a shorthand for $\neg\T$.

The set of (free) variables of a term $t$ is denoted by $\fv(t)$; notice that there are no binding operators in terms. As usual, a term $t$ is said to be \emph{ground} if $\fv(t) = \emptyset$, i.e. without (free) variables. It is said to be a \emph{message} if it is ground and without occurrences of $\pi_1(.)$, $\pi_2(.)$ and $\D{.}{.}$ operators. 
We will denote with $\M$ the set of all messages, ranged over by $M,N$.

Unlike \cite{ag:spi, huttel:infinity02}, our terms are typed.
Types are defined by the following syntax:
\begin{equation*}
\tau ::= \typeN \mid \typeB \mid \tau_1 \times \tau_2 \mid \typeC(\tau)
\end{equation*}
where $\typeN, \typeB$ are the base types of names and booleans respectively, and $\typeC(\tau)$ is the type of encrypted terms of type $\tau$.
Formally, the typing judgment $t:\tau$ over ground terms is defined by the following rules.
\begin{gather*}
\inference{a\in\N}{a:\typeN}\quad
\inference{}{\T:\typeB}\quad 
\inference{\phi:\typeB}{\lnot\phi:\typeB}\quad
\inference{\phi_1:\typeB\quad\phi_2:\typeB}{\phi_1\land\phi_2:\typeB}\quad
\inference{}{[\term_1=\term_2]:\typeB}
\\
\inference{\term_1:\tau_1\quad\term_2:\tau_2}{(\term_1,\term_2):\tau_1\times\tau_2}\quad
\inference{\term:\tau_1\times\tau_2}{\pi_i(\term):\tau_i} i=1,2
\quad
\inference{\term_1:\tau\quad \term_2:\typeN}{\E{\term_1}{\term_2}:\typeC(\tau)}\quad
\inference{\term_1 : \typeC(\tau) \quad \term_2: \typeN}{\D{\term_1}{\term_2}:\tau}
\end{gather*}

\paragraph{Congruence over terms}
Terms are taken up-to some structural congruence $\equiv$, whose aim is to express the evaluation internal to processes, in particular the execution of encryption/decryption algorithms.
Differently from \cite{ag:spi,bn:mscs05,huttel:infinity02}, we aim to account for different type of encryption algorithms that can be expressed by choosing the structural equivalence $\equiv$. 
To this end, we provide a general definition of ``coherent'' congruence:
\begin{definition}[Coherent Congruence]\label{def:cong}
A congruence relation $\equiv$ over terms is \emph{coherent} if the following hold:
\begin{enumerate}
	\item (Type preservation) $\forall t_1:\tau_1, t_2: \tau_2 \forall j,k \in \N: \E{t_1}{k} \equiv \E{t_2}{j} \Rightarrow \tau_1 = \tau_2$
	\item (Equivariance) $\forall t_1 , t_2 : \typeC(\tau)\ \forall a \in \N\ \forall b \in \N \setminus (n(t_1) \cup n(t_2)): t_1 \equiv t_2 \iff t_1\{b/a\} \equiv t_2\{b/a\}$
	\item (Deterministic decryption) $\forall t_1 , t_2 : \typeC(\tau):\ t_1 \equiv t_2 \Rightarrow \ds(t_1) = \ds(t_2)$
\end{enumerate}
where $n(t)$ is the set of all names occurring in $t$,
$t\{s/a\}$ is the syntactical substitution replacing all occurrences of $a$ in $t$ with the term $s$, and
$\ds(.)$ is defined inductively by the clauses
\[
\ds(n \in \N) = n\quad
\ds((t_1,t_2)) = (\ds(t_1),\ds(t_2))\quad
\ds(\E{t_1}{t_2}) = \ds(t_1)
\]
\end{definition}
The first condition says that two encrypted terms can be considered equal only if they encrypt messages with the same type.
The second condition imposes the absence of special names (and keys): if a property holds for a name, then it must hold for every fresh name.
The third condition says that congruence must be consistent with decryption: the decryption of a message $M$ is thus guaranteed to be deterministic w.r.t. all messages in the same equivalence class of $M$.

It is easy to check that the equivalence used in \cite{ag:spi,bn:mscs05,huttel:infinity02} respects these conditions. Other encryption algorithms (and other abstract data types) can be considered by adapting the congruence relation, as long as it remains coherent; for example, we can analyze encryption protocols with commutative ciphers (like RSA) by adding the axiom $\forall M :\tau\ \forall k,j \in \N:\ \E{\E{M}{k}}{j} \equiv \E{\E{M}{j}}{k}$.

\paragraph{Processes}
We can now define the processes of the spi-calculus.
\begin{definition}[Processes]
The spi-calculus  processes are defined as follows:
\begin{equation*}
P ::=\ 0 \mid A(x).P \mid \overline{A}\langle \term \rangle.P \mid P_1|P_2 \mid P_1 + P_2 \mid (\nu a)P  \mid !P 
\mid \phi.P \mid \text{let } x = \term \text{ in } P
\end{equation*}
where $t$ and $\phi$ are respectively a well-typed term and a boolean formula, $x$ is a variable, $a$ is a name and $A$ can be both a name or a variable.  
\end{definition}
The syntax above are the usual ones from $\pi$-calculus, with these differences: 
\begin{itemize}
	\item input/output operations exchange terms, not only names and variables;
	\item $\phi.P$ is the \emph{guard operator}, that behaves as $P$ if the boolean formula $\phi$ holds;
	\item $\textit{let } x = \term \textit{ in } P$ is the \emph{let operator} that computes the value of $\term$, assigns it to the variable $x$ and then executes $P$.
\end{itemize}

Without loss of generality, we can assume that destructors ($\pi_1(.)$, $\pi_2(.)$ and $\D{.}{.}$) do not occur in boolean predicates $[t_1=t_2]$, nor in the argument of output operations $\overline{A}\langle \term \rangle.P$, since these cases can be simulated using the \emph{let}.
For instance, $\overline{a}\langle \pi_1(t) \rangle.P$  is equivalent to $\textit{let } x=\pi_1(t) \textit{ in }  \overline{a}\langle x \rangle.P$, for $x\not\in \fv(P)$.

Processes are taken up-to the usual structural congruence familiar from the $\pi$-calculus theory:
\begin{align*}
P &\equiv Q   \hspace{1.5cm} \makebox[2cm]{ if $P$ and $Q$ are $\alpha$-equivalent} &
!P &\equiv P | !P
\\
P|Q &\equiv Q|P & 
(P|Q)|R &\equiv P(Q|R)
\\
P + Q &\equiv Q + P &
(P + Q) + R &\equiv P + (Q + R) 
\\
(\nu m)(P|Q) &\equiv (\nu m)P | Q \text{ if } m \not\in\fn(Q) &
(\nu m)(\nu n)P &\equiv (\nu n)(\nu m)P 
\\
(\nu m) 0 &\equiv 0 \quad
P | 0 \equiv P 
\quad
\inference{P \equiv Q}{P|R \equiv Q|R} & 
\inference{P \equiv Q}{P+R \equiv Q+R} &\quad
\inference{P \equiv Q}{(\nu m) P \equiv (\nu m) Q}
\end{align*}

\subsection{Semantics} 
To define the operational semantics of the spi-calculus we need to evaluate terms and boolean expressions.
Evaluation is defined over \emph{well-typed} terms, where each type denotes a set of values:
\begin{definition}[Interpretation of types]\label{def:typesem}
	The interpretation of types $\sd{\cdot}:\text{Types} \to \text{Set}$ is defined recursively as follows:
\begin{align*}
\sd{\typeN} &= \N  \\
\sd{\typeB} &= \{\T,\F\} \\
\sd{\typeC(\tau)} &= \{\E{M}{k} \mid M:\typeC(\tau), k\in\N \}/{\equiv} \\
\sd{\tau_1\times\tau_2} &= \{ (M_1,M_2) \mid M_1\in \sd{\tau_1}, M_2\in\sd{\tau_2}\}/{\equiv}
\end{align*}
\end{definition}

\begin{definition}[Evaluation]
The \emph{evaluation} for ground terms and boolean expressions is a partial function
$\sd{\cdot}:M_\tau \rightharpoonup \sd{\tau}$
(implicitly parametric in the type $\tau$) defined recursively as follows:
\begin{align*}
\sd{a} &= a \in \N
& 
\sd{(\term_1,\term_2)} &= (\sd{\term_1} , \sd{\term_2})
\\
\sd{\T} &= \T
&
\sd{\pi_1(\term)} &= v_1 \text{ if } \sd{\term} = (v_1,v_2)
\\
\sd{\phi_1\land\phi_2} &= \sd{\phi_1} \land \sd{\phi_2}
&
\sd{\pi_2(\term)} &= v_2  \text{ if } \sd{\term} = (v_1,v_2)
\\ 
\sd{\lnot\phi} &= \lnot\sd{\phi}
&
\sd{\E{\term_1}{\term_2}} &= \E{\sd{\term_1}}{\sd{\term_2}} %
\\
\sd{[\term_1 = \term_2]} &= \T \hspace{1.8cm} \makebox[1cm]{if $\sd{\term_1} \equiv \sd{\term_2}; \F$ otherwise}
\\
\sd{\D{\term_1}{\term_2}} & = t \in \tau \hspace{3cm}
  \makebox[1cm]{ if $k = \sd{\term_2} \in \N$ and $\sd{\term_1} \equiv \E{t}{k} \in \sd{\typeC(\tau)}$}
\end{align*}
We will write $\sd{t}\!\downarrow$ if the evaluation of $t$ produces a value, $\sd{t}\!\uparrow$ otherwise -- i.e. when the evaluation of a decryption or a projection is unsuccessful.
\end{definition}

Following \cite{ag:spi,huttel:infinity02} (and differently from \cite{bn:mscs05}) we define a \emph{late input} style operational semantics. 
We first define the \emph{reduction} relation, which describes how processes unfold and execute internal computations in preparation for a reaction. 
\[
\inference{\sd{\phi} = \T}{\phi.P > P} 
\quad
\inference{\sd{\term}=v}{\text{let } x = \term \text{ in } P > P\{v/x\}}
\quad
\inference{P > P'}{P \equiv P'}
\]
where $v$ is the value of $t$ (up-to congruence), if defined.

The next step is to define \emph{abstractions} $F$ and \emph{concretions} $C$:
\[
F ::= (x)P
\qquad
C ::= (\nu m_1\dots m_n)\langle M \rangle Q \quad n\geq 0
\]
where the variable $x$ is bound in $P$ and names $m_1,\dots,m_n$  are bound in $M,Q$.
As we will see in the semantic rules, an input $a(x).P$ becomes an abstraction after performing a transition labeled $a$; this abstraction can be seen as a process waiting to receive a message on channel $a$.
An output $\overline{a}\langle M \rangle.Q$ becomes a concretion $(\nu \vec{m})\langle M \rangle Q$, where $\vec{m}$ are fresh names that can appear in $M$ and $P$. This concretion can be seen as a process waiting to send a message on the channel $a$.

An abstraction $(x)P$ and a concretion $(\nu \vec{m})\langle M \rangle Q$ can interact via synchronization resulting in a process where the message $M$ is received by $(x)P$.
In order to define this interaction, we need to extend restriction and parallel composition operators to abstractions and concretions, as follows:
\begin{align*}
(\nu m)(x)P &\triangleq (x)(\nu m)P\\
(\nu n)(\nu \vec{m})\langle M \rangle P &\triangleq \begin{cases}
(\nu n, \vec{m})\langle M \rangle P &\text{if } n \in \fn(M)\\
(\nu \vec{m})\langle M \rangle (\nu n) P &\text{otherwise}
\end{cases}\\
R | (x)P &\triangleq (x)(R|P) \quad \text{where } x \not\in \fv(R)\\
R | (\nu\vec{m})\langle M \rangle P &\triangleq (\nu\vec{m})\langle M \rangle (R|P) \quad \text{where } \{\vec{m}\} \cap \fn(R) = \emptyset
\end{align*}
where the two last definitions can be always applied, by $\alpha$-conversion.

Finally, the operational semantics of the spi-calculus is represented by a labelled transition relation $P\overset{\alpha}{\longrightarrow} D$, where $D$ ranges over processes, concretions and abstractions, and $\alpha\in \{a,\overline{a} \mid a\in \N \}\cup\{\tau\}$ is the label.
As usual, the transition labelled with $\tau$ is also called silent transition or $\tau$-transition.
The relation is defined by the rules given in Figure~\ref{fig:lts}.
As usual, we also define
\[
P \Longrightarrow Q \overset{\triangle}{\iff} P\overset{\tau}{\longrightarrow}^*Q
\qquad\qquad
P \overset{\alpha}{\Longrightarrow} Q \overset{\triangle}{\iff} P \Longrightarrow \overset{\alpha}{\longrightarrow} Q
\]

\begin{figure}
\begin{gather*}
\inference[(input)]{a\in \N}{a(x) P \overset{a}{\longrightarrow} (x)P}\quad\quad
\inference[(output)]{\sd{t}=M \quad a\in\N}{\overline{a}\langle\term\rangle P \overset{\overline{a}}{\longrightarrow} (\nu)\langle M\rangle P}
\\
\inference[(interaction)]{P \overset{n}{\longrightarrow} (x)P'\quad 
	Q \overset{\overline{n}}{\longrightarrow} (\nu \vec{m})\langle M \rangle Q'
	}{P|Q \overset{\tau}{\longrightarrow} (\nu \vec{m})(P'\{M/x\}|Q')}\{\vec{m}\} \cap \fn(P') = \emptyset
\\
\inference[(restriction)]{P \overset{\alpha}{\longrightarrow} D \quad\alpha \not\in \{m,\overline{m}\}}{(\nu m) P \overset{\alpha}{\longrightarrow} (\nu m) D}\quad\quad
\inference[(parallel)]{P \overset{\alpha}{\longrightarrow} P'}{P|Q \overset{\alpha}{\longrightarrow} P' | Q}
\\
\inference[(sum)]{P \overset{\alpha}{\longrightarrow} P'}{P + Q \overset{\alpha}{\longrightarrow} P'}\quad\quad
\inference[(equivalence)]{P \equiv Q\quad Q \overset{\alpha}{\longrightarrow} Q'\quad Q' \equiv P'}{P \overset{\alpha}{\longrightarrow} P'}
\end{gather*}
\caption{Late operational semantics of the spi-calculus.}
\label{fig:lts}
\end{figure}

\section{Hedged bisimilarity}\label{sec:hedged}
In this section we define the \emph{hedged bisimilarity}, introduced in \cite{bn:mscs05}.
The basic idea is to mimic the \emph{frame-theory} pairs of the \emph{framed bisimilarity} defined in \cite{AbadiG98}, but dropping the separate \emph{frame} component and including corresponding names as part of the \emph{theory}. The resulting theory is then called a \emph{hedge}.
\begin{definition}[Hedge]
Let $\M$ be a set of messages. A hedge is a finite subset of $\M^2$. We denote by $\Hedges$ the set of all hedges.

A hedge $h$ is \emph{consistent} if and only if it is \emph{pair-free} (i.e. all messages in $h$ are not pairs) and whenever $(M,N) \in h$ we have that:
\begin{itemize}
	\item $M \in \N \iff N \in \N$;
	\item for all $(M',N') \in h$, if $M \equiv M'$  or $N \equiv N'$ then $M=M'$ and $N=N'$; 
	\item if $M \equiv \E{M'}{k}$ e $N \equiv \E{N'}{j}$ then $k \not\in \pi_1(h)$ and $j \not\in \pi_2(h)$.
\end{itemize}
\end{definition}
The first condition requires a consistent hedge to match names with names.
The second one requires that $\pi_1(h)$ and $\pi_2(h)$ are in bijection.
The last one requires that all encrypted messages cannot be decrypted using keys in $h$: encrypted messages in a consistent hedge are not reducible.

Alongside (consistent) hedges, we define \emph{synthesis}, \emph{analysis}, \emph{irreducibles}.
The \emph{synthesis} of a hedge $h$ is the set of message pairs that can be built from $h$.
\begin{definition}[Synthesis] The synthesis $S(h)$ of a hedge $h$ is defined as the least subset of $\M^2$ that satisfies:
	\begin{itemize}
		\item $h \subseteq S(h)$;
		\item if $(M,N) \in S(h)$, $(k,j) \in S(h)$ and $k,j \in \N$ then $(\E{M}{k},\E{N}{j}) \in S(h)$;
		\item if $(M_1,N_1) \in S(h)$ and $(M_2,N_2) \in S(h)$ then $((M_1,M_2),(N_1,N_2)) \in S(h)$.
	\end{itemize}
	We write $h \vdash M \leftrightarrow N$ for $(M,N) \in S(h)$, and in this case we say that $M$ and $N$ are \emph{homologous} w.r.t. $h$.
 \end{definition}

The \emph{analysis} of a hedge is the set of all message pairs obtained by ``opening" the messages of $h$ via decryption or projection.
The \emph{irreducibles} are those elements in the analysis of a hedge that cannot be reduced further.
Formally:
\begin{definition}[Analysis] The analysis $A(h)$ of a hedge $h$ is defined as the least set that satisfies:
\begin{itemize}
	\item $h \subseteq A(h)$;
	\item if $(\E{M}{k}, \E{N}{j}) \in A(h)$ and $(k,j) \in A(h)$ then $(M,N) \in A(h)$;
	\item if $((M_1,N_1),(M_2,N_2)) \in A(h)$ then $(M_1,M_2) \in A(h)$ and $(N_1,N_2) \in A(h)$.
\end{itemize}
Moreover, we define the \emph{irreducibles} $I(h)$ of a hedge $h$ as 
\begin{align*}
I(h) \triangleq A(h) \setminus (&\{ (C,D) \in A(h) \mid C \equiv \E{M}{k},\ D \equiv \E{N}{j},\ (k,j) \in A(h) \}\\ &\cup \{ ((M_1,N_1),(M_2,N_2)) \in A(h)\})
\end{align*}
\end{definition}
It should be noted that all elements that can be reduced in the analysis can be derived from the irreducibles via synthesis, i.e. $S(I(h)) = S(A(h))$.
Lastly, since every hedge is a finite set, its analysis and irreducibles are also finite.

\paragraph{Hedge simulations}
Let us recall that $\Hedges$ and $\Prt$ are the set of all hedges and the set of all processes, respectively. A \emph{hedged relation} $\R$ is a subset of $\Hedges \times \Prt \times \Prt$. We write $h \vdash P \R Q$ when $(h,P,Q) \in \R$. Moreover, we say that $\R$ is \emph{consistent} if, for all $h \in H$, $h \vdash P \R Q$ implies that the hedge $h$ is consistent.

\begin{definition}[Hedged simulation]\label{def:hs}
A consistent hedged relation $\R$ is a \emph{hedged simulation} if, whenever $h \vdash P \R Q$ we have that:
\begin{itemize}
	\item if $P \overset{\tau}{\longrightarrow} P'$ then there exists $Q'$ such that $Q \Longrightarrow Q'$ and $h \vdash P' \R Q'$;
	\item if $P \overset{\overline{a}}{\longrightarrow} (\nu \vec{m})\langle M \rangle P'$ and $\{\vec{m}\} \cap (\fn(P) \cup n(\pi_1(h))) = \emptyset$ then
	there exist $b\in \N$ and a concretion $(\nu \vec{n}) \langle M \rangle Q'$ such that
	$h \vdash a \leftrightarrow b$,
	$\{\vec{n}\} \cap (\fn(Q) \cup n(\pi_2(h))) = \emptyset$,
	$Q \overset{\overline{b}}{\Longrightarrow}(\nu \vec{n}) \langle M \rangle Q'$ and $I(h\cup\{(M,N)\}) \vdash P' \R Q'$;
	\item if $P \overset{a}{\longrightarrow} (x)P'$ then
	there exist $b \in \N$ and an abstraction $(y)Q'$ such that
	$h \vdash a \leftrightarrow b$,
	$Q \overset{b}{\Longrightarrow} (y)Q'$ and
	for all $B \subset \N$ finite such that $B \cap (\fn(P) \cup \fn(Q) \cup n(h)) = \emptyset$ and $h \cup id_B$ is consistent, for all pairs $(M,N)$ of ground terms, if $h \cup id_B \vdash M \leftrightarrow N$ then $h \cup id_B \vdash P'\{M/x\} \R Q'\{N/y\}$.
\end{itemize}
\end{definition}
The first condition requires that for each $\tau$-transition from $P$ there is a path of $\tau$-transition from $Q$ such that the two target processes are in the simulation $\R$.
The second condition requires that for each output transition of $P$, labelled with $\overline{a}$, there is an output transition from $Q$ labelled with $\overline{b}$ (and possibly preceded by some silent transitions); moreover, $a$ and $b$ are homologous in $h$ and the processes after the two output operations are paired in $R$ w.r.t. a consistent hedge that extends $h$ by pairing the two messages $M$ and $N$.
The last condition requires that for each input transition of $P$ with label $a$, there is an input transition from $Q$ labelled with $b$ (and possibly preceded by some silent transitions); moreover, $a$ and $b$ are homologous in $h$ and for all finite set $B$ of fresh names w.r.t. $P$, $Q$ and $h$, the abstractions $(x)P'$ and $(x)Q'$ are paired in the simulation $R$ for each input messages $(M,N)$ homologous by $h \cup id_B$.

\begin{definition}[Hedged bisimulation and bisimilarity]
	A hedged simulation $\R$ is a \emph{hedged bisimulation} if $\R^{-1} = \{(h^{-1},Q,P) \mid h \vdash P \R Q\}$  is also a hedged simulation (where ${h^{-1} = \{ (N,M) | (M,N) \in h\}}$).

\emph{Hedged bisimilarity}, written $\sim$, is the greatest hedged bisimulation, i.e. the union of all hedged bisimulations.
\end{definition}

Remarkably, hedged bisimilarity coincides with barbed bisimilarity \cite{bn:mscs05}.

\section{Decidability of hedged bisimulation for finite processes}\label{sec:decid}
Definition~\ref{def:hs} does not provide us with a means for checking bisimilarity.
In this Section we address this issue, following, when possible, the approach in \cite{huttel:infinity02}.

Clearly, bisimilarity is undecidable for general, infinite processes; hence, we focus on \emph{finite} processes, i.e. without replication.
Even on finite processes decidability of hedged bisimilarity is not obvious, since the third condition in Definition \ref{def:hs} requires to check the equivalence of two abstractions for an infinite number of messages w.r.t. any finite set of fresh names.
In this section we prove that there is a finite bound on the number of these names and messages. If this bound exists, then the hedged bisimilarity is trivially decidable.

The idea behind our result, as in \cite{huttel:infinity02}, is the following: if $(x)P$ is finite, then it can inspect a message (using \emph{let} and \emph{guard} operators) up to a certain depth $k$.
If a message $M$ with more than $k$ nested constructors is received by $(x)P$, then it can only be partially analysed by $P$.
Hence, all messages $M'$ equivalent to $M$ up to depth $k$ will not cause any difference in the execution of $(x)P$, apart from output messages. Indeed, $P\{M/x\}$ and $P\{M'/x\}$ can output different messages (i.e. different parts of $M$ and $M'$ respectively), but we notice that:
\begin{itemize}
	\item the two outputs are derived from $M$ and $M'$ by applying the same operations;
	\item only messages obtained through decryption are interesting, since they can update the hedge $h$ yielding a richer theory.
\end{itemize}

We now proceed to formalize this idea.

\begin{definition}[Maximal constructor depth]
The \emph{maximal constructor depth} $\cc(M)$ of a message $M$ is defined inductively by the clauses
\begin{align*}
\cc(n \in \N) &= 0 &
\cc((M_1,M_2)) &= \max( \cc(M_1), \cc(M_2)) + 1\\
\cc(x \in \V) &= 0 &
\cc(\E{V}{K}) &= \cc(V) + 1
\end{align*}
and then extended to boolean formulas as follows:
\begin{align*}
\cc(\T) &= 0 &
\cc(\phi_1 \land \phi_2) &= \max(\cc(\phi_1),\cc(\phi_2))\\
\cc(\lnot\phi) &= \cc(\phi) &
\cc([M = N]) &= \max(\cc(M),\cc(N))
\end{align*}
\end{definition}

\begin{definition}[$k$-homologous]
Given $h\in\Hedges$ and  $M,N\in\M$, we define
\vspace{-1ex}
\begin{equation*}
h \vdash_k M \leftrightarrow N \overset{\triangle}{\iff} h \vdash M \leftrightarrow N \text{ and } k = \max(\cc(M),\cc(N))
\end{equation*}
Whenever $h \vdash_k M \leftrightarrow N$ we say that $M$ and $N$ are \emph{$k$-homologous} in $h$.
\end{definition}

The notion of \emph{maximal constructor depth} is readily extended to hedges:
\vspace{-1ex}
\begin{equation*}
\cc(h) \defeq \max\{k \mid \exists (M,N) \in h:\ h \vdash_k M \leftrightarrow N \}
\end{equation*}

\begin{lemma}\label{lem:namebound}
	Let $h\in\Hedges$ and $M, N\in \M$ be such that $\max(\cc(M),\cc(N)) = k$. If there is a finite set of names $B \subset \N$ such that
	\begin{itemize}
		\item $B \cap n(h) = \emptyset$;
		\item $h \cup id_B$ is consistent;
		\item $h \cup id_B \vdash_k M \leftrightarrow N$
	\end{itemize}
	then there exists $B' \subset \N$ with the same properties and such that $|B'| \leq 2^k$.
\end{lemma}
\begin{proof}
	If $|B| \leq 2^k$ or $h \vdash M \leftrightarrow N$, then the thesis follows trivially and $B' = B$ or $B=\emptyset$, respectively.
	Otherwise, $M$ and $N$ are in the synthesis $S(h \cup id_B)$ and, at worst, all names in $M$ and $N$ are in $B$.
	Since $k = \max(\cc(M),\cc(N))$ and both \emph{encrypt} and \emph{pairing} are binary constructors, $M$ and $N$ can be represented as binary trees with height $k$. A binary tree of height $k$ has at most $2^k$ leafs, hence $B$ can be reduced to a set $B'$ such that $|B'| \leq 2^k$ without losing any property in the hypothesis.
	\qed
\end{proof}

Lemma~\ref{lem:namebound} leads us to the definition of \emph{d-hedged bisimulation}: a hedged bisimulation up to a bound $d$ on message depth and a bound $2^d$ on fresh names.

\begin{definition}[d-hedged simulation]\label{def:dsim}
For any integer $d \geq 0$, a consistent hedged relation $\R$ is a \emph{$d$-hedged simulation} if whenever $h \vdash P \R Q$ we have that:
	\begin{itemize}
		\item if $P \overset{\tau}{\longrightarrow} P'$ then there exists $Q'$ such that $Q \Longrightarrow Q'$ and $h \vdash P' \R Q'$;
		\item if $P \overset{\overline{a}}{\longrightarrow} (\nu \vec{m})\langle M \rangle P'$ and $\{\vec{m}\} \cap (\fn(P) \cup n(\pi_1(h))) = \emptyset$ then
		there exist $b\in \N$ and a concretion $(\nu \vec{n}) \langle M \rangle Q'$ such that 
		$h \vdash a \leftrightarrow b$,
		$\{\vec{n}\} \cap (\fn(Q) \cup n(\pi_2(h))) = \emptyset$,
		$Q \overset{\overline{b}}{\Longrightarrow}(\nu \vec{n}) \langle M \rangle Q'$
		and $I(h\cup\{(M,N)\}) \vdash P' \R Q'$;
		\item if $P \overset{a}{\longrightarrow} (x)P'$ then
		there exist $b \in \N$ and an abstraction $(y)Q'$ such that
		$h \vdash a \leftrightarrow b$,
		$Q \overset{b}{\Longrightarrow} (y)Q'$ and
		for all $B \subset \N$, where $|B| \leq 2^d$, $B \cap (\fn(P) \cup \fn(Q) \cup n(h)) = \emptyset$ and $h \cup id_B$ is consistent, for all pairs $(M,N)$ of ground terms, if $\exists k \leq d\ h \cup id_B \vdash_k M \leftrightarrow N$ then $h \cup id_B \vdash P'\{M/x\} \R Q'\{N/y\}$.
	\end{itemize}
\end{definition}

\begin{definition}[d-hedged bisimulation and bisimilarity]
	A \emph{$d$-hedged bisimulation} is a $d$-hedged simulation $\R$ such that $\R^{-1} = \{(h^{-1},Q,P) \mid h \vdash P \R Q\}$, where ${h^{-1} = \{ (N,M) | (M,N) \in h\}}$, is also a d-hedged simulation.

	The \emph{d-hedged bisimilarity}, written $\sim^d$, is the greatest d-hedged bisimulation, i.e. the union of all $d$-hedged bisimulations.
\end{definition}

These definitions immediately lead to the following results.
\begin{proposition}\label{prop:dedged}
\begin{enumerate}[(a)]
\item Every hedged bisimulation is also a d-hedged bisimulation, for any $d\geq 0$.
\item For any $d>0$, a d-hedged bisimulation is also a $(d-1)$-hedged bisimulation.
\end{enumerate}
\end{proposition}
\begin{proof}
\begin{enumerate}[(a)]
\item  By removing the cardinality constraint on the set $B$ and requiring only its finiteness, we get the definition of hedged simulation; hence hedged simulations satisfy the definition of $d$-hedged simulation for any $d$.
\item Let $h$ be a hedge and $d > 0$; it is trivial to check that if $h \vdash P \sim^d Q$ then $h \vdash P \sim^{d-1} Q$. 
\qed
\end{enumerate}
\end{proof}

We now aim to show that, for any two processes $P,Q$, there exists $d \geq 0$ such that $\exists h \in \Hedges\ h \vdash P \sim^d Q \Rightarrow \exists h \in \Hedges\ h \vdash P \sim Q$.
This does not hold for arbitrary infinite processes since these can analyse messages of arbitrary depth.
Therefore, we now consider only the fragment of spi-calculus without replication.

It should be noted that \emph{let} and \emph{guard} operators are the only constructs that can check the structure of messages.  For instance, the process $c(x).[x=t].P$ uses the term $t$ to ``test'' a message received along channel $c$; therefore, any message with depth greater than $t$'s will fail the test, and hence we can consider to send the process only messages with depth up-to that of $t$.
This observation leads to the following definition of  \emph{analysis depth}.

\begin{definition}[Analysis depth]\label{def:ad}
Let $P$ be a finite process. The \emph{analysis depth} of $P$, denoted by $\coa(P)$, is defined inductively by the clauses
\begin{align*}
\coa(0) &= 0 & 
\coa(P|Q) &= \coa(P) + \coa(Q) \\
\coa((\nu n) P) &= \coa(P) & 
\coa(\phi.P) &= \max(\coa(P), \cc(\phi)) \\
\coa(\overline{M}\langle N \rangle.P) &= \coa(P) &
\coa(P+Q) &= \max(\coa(P),\coa(Q)) \\
\coa(M(x).P) &= \coa(P) &
\coa(\text{let } x = t \text{ in } P) &= \coa(P\{t'/x\}) + \mdd(t)
\end{align*}
where, in the case of the \emph{let} operator, $t'$ is any message such that $\cc(t') = \cc(t)$ (e.g., $t'$ can be obtained by nesting $\cc(t)$ encryptions with a fresh name) and the \emph{maximal destructor depth} of a term is defined as follows:
\begin{align*}
\mdd(n) &= 0 &
\mdd(\E{t_1}{t_2}) &= \max(\mdd(t_1), \mdd(t_2))\\
\mdd(x) &= 0 &
\mdd((t_1,t_2)) &= \max(\mdd(t_1), \mdd(t_2))\\
\mdd(\pi_i(t)) &= \mdd(t) + 1 \quad i =1,2  &
\mdd(\D{t_1}{t_2}) &= \max(\mdd(t_1), \mdd(t_2)) + 1.
\end{align*}
\end{definition}

\begin{definition}[Critical depth]
	Let $P$ and $Q$ be two finite processes and let $h$ be a hedge. The \emph{critical depth} $\cri(h,P,Q)$ is defined by
	\begin{equation*}
		\cri(h,P,Q) \triangleq \cc(h) + \max(\coa(P), \coa(Q)) + 1
	\end{equation*}
\end{definition}

\begin{remark}\label{rem1}
	In the definition of analysis depth (Definition~\ref{def:ad}) we have taken into account also the analysis done by matching operators.
	This differs from H\"uttel's work about framed bisimilarity, where in the definition of the analysis depth (there called ``maximal destruction depth'' \cite[Def.~13]{huttel:infinity02})  it is $\coa([M = N].P) = \coa(P)$.
	In fact, it is crucial to consider also the matching operator. As an example, let us consider the following processes and frame-theory pair:
	\begin{align*}
	P &= a(x).[x = \E{a}{a}].\bar{a}\langle a \rangle.0	&
	Q &= a(x).0  &
	(fr,th) &= (\{a\},\emptyset)
	\end{align*}
	According to \cite{huttel:infinity02} the critical depth would be $\cri((fr,th),P,Q) = 0$; hence $\E{a}{a}$ would not be considered as a possible input message, since $\cc(\E{a}{a}) = 1$ and $(fr,th) \not\vdash_0 \E{a}{a} \leftrightarrow \E{a}{a}$.
	Therefore $P$ and $Q$ would behave similarly for each input message tested by H\"uttel's algorithm, leading to incorrectly conclude that $(fr,th) \vdash P \sim Q$.
	This does not happen if the analysis depth takes into account the number of constructors used by matching operators, as in Definition~\ref{def:ad}.
	\qed
\end{remark}

As we will formally see below, when checking if a $d$-hedged simulation exists, we can correlate two input $P \overset{a}{\longrightarrow} (x)P'$ and $Q \overset{b}{\longrightarrow} (x)Q'$ w.r.t. a hedge $h$, simply by testing the equivalence between $P'$ and $Q'$ w.r.t. messages with \emph{maximal constructor depth} less or equal to  $\cri(h,P,Q)$.
Moreover, Lemma~\ref{lem:namebound} limits the number of fresh names we must consider to  $2^{\cri(h,P,Q)}$.

We will now prove that, if $P,Q$ are finite processes and $d = \cri(h,P,Q)$ then 
\begin{equation*}
\exists h \in \Hedges\ h \vdash P \sim^d Q\ \Longrightarrow\ \exists h \in \Hedges\ h \vdash P \sim Q.
\end{equation*}
This result is based on the definition of \emph{d-pruning} and two lemmata that show the equivalence of the transition system when considering only messages $M$ such that $\cc(M) \leq \cri(h,P,Q)$.

\begin{definition}[d-pruning]
	Let $M$ and $N$ be two messages and $h$ a consistent hedge such that $h \vdash M \leftrightarrow N$.
	For $d\geq 0$, the \emph{d-pruning} of $M$ and $N$ w.r.t. $h$, denoted by $\pr_d(h,M,N)$ is defined by cases as follows:
	\begin{itemize}
	\item if $(M,N) \in h$ then for all $d: \pr_d(h,M,N) \defeq (h,M,N)$
	\item if $(M,N) \not\in h$ then 		
	\begin{align*}
		\pr_0(h,M,N) &= (h \cup id_{\{a\}}, a, a) \text{ where } a \in \N \text{ is fresh}
		 \\
		\pr_{d+1}(h,\E{U}{J},\E{V}{K}) &= (h', \E{M'}{J}, \E{N'}{K}) \\
		& \qquad \text{where } (J,K) \in h \\
		& \qquad \text{and } \pr_d(h,U,V) = (h',M', N')
		\\
		\pr_{d+1}(h, (L_1,R_1), (L_2,R_2)) &= (h^{\prime\prime}, (L'_1,R'_1), (L'_2, R'_2))\\
		&\qquad \text{where } \pr_d(h, L_1, L_2) = (h', L'_1, L'_2)\\
		&\qquad \text{and } \pr_d(h', R_1, R_2) = (h'', R'_1, R'_2).
	\end{align*}
	\end{itemize}
\end{definition}
Intuitively, the $d$-pruning of a message pair $(M,N)$ generates a message pair $(M',N')$ where subterms appearing at levels greater than $d$ are replaced by fresh names w.r.t. $h$, $M$ and $N$.

Critical depth and $d$-pruning are readily extended to processes and messages:
\begin{align*}
	\cri(h,P) &\triangleq \cri(h,P,P) &
	\pr_d(h,M) &\triangleq \pr_d(h,M,M)
\end{align*}

\begin{lemma}\label{lem:reductions}
	Let $(x)P$ be an abstraction of a finite process, $h$ be a hedge and $d = \cri(h,P)$. For every message $M$,
	\begin{equation*}
		P\{M/x\} > P'_M\{M/x\} \iff P\{N/x\} > P'_N\{N/x\} 
	\end{equation*}
	where the same reduction rule is used, $N = \pi_2(\pr_d(h,M))$ and 
	\subitem in the case of the guard reduction: $P'_M = P'_N$
	\subitem  in the case of the \emph{let} reduction:  $P'_M = Q\{\sd{t\{M/x\}}/y\}$ and $P'_N = Q\{\sd{t\{N/x\}}/y\}$, for some $Q$.
\end{lemma}
\begin{proof}
		Let us consider the case of \emph{guard} reduction, i.e., $P\{M/x\} = (\phi.Q)\{M/x\}$ for some $Q$ and 
		$\sd{\phi\{M/x\}} = \T$. 
		We only need to show that
		\begin{equation*}
		\sd{\phi\{M/x\}} \iff \sd{\phi\{N/x\}}
		\end{equation*}
		We proceed by induction on the structure of $\phi$. The inductive cases are easy:
		\begin{align*}
		\sd{\lnot\phi\{M/x\}} &\iff \lnot \sd{\phi\{M/x\}} \overset{IH}{\iff} \lnot \sd{\phi\{N/x\}} \iff \sd{\lnot\phi\{N/x\}}\\
		\sd{(\phi \land \psi) \{M/x\}} &\iff \sd{\phi\{M/x\}} \land \sd{\psi\{M/x\}}  \\
		 & \overset{IH}{\iff}  \sd{\phi\{N/x\}} \land \sd{\psi\{N/x\}} \iff \sd{(\phi\land\psi)\{N/x\}}
		\end{align*}

		Let us consider the base case of matching predicate $[t_1 = t_2]$.
		If $\fv(t_1) \cup \fv(t_2) = \emptyset$ then the result of the matching is independent from the message and the lemma trivially holds. Otherwise, i.e. $\fv(t_1) \cup \fv(t_2) = \{x\}$, the matching $[t_1 = t_2]\{M/x\}$ is true if and only if $\sd{t_1\{M/x\}} \equiv \sd{t_2\{M/x\}}$.
		Moreover, the first condition of Definition~\ref{def:cong} requires $t_1\{M/x\}$ and $t_2\{M/x\}$ to have the same type. Therefore, for each path in the syntactic tree of $t_1$ ending with an occurence of $x$, there is an equivalent path in $t_2$ ending with either an occurrence of $x$ or a message of the same type of $M$; and vice versa for the occurrences of $x$ in $t_2$. Hence, without loss of generality we can limit ourselves to matchings of the form $[x = x]$ (which is trivial) and $[x = t]$---which we show next. 

		If $\cc(M) \leq d$ then $\pr_d(M) = M$ and the lemma is proved.
		If $\cc(M) > d$ then $\cc(M) \geq \cc(N) \geq d > \cc(t)$, since $d$ is the depth of the pruning. Therefore it holds that $\sd{[x = t]\{M/x\}} = \sd{[x = t]\{N/x\}} = \F$ and the lemma is proved. Moreover, $P'_M = P'_N$.
		
		Let us consider now the \emph{let} reduction, i.e., $P\{M/x\} = (\text{let } y = t \text{ in } Q)\{M/x\}$ for some $t,Q$.
		If $M$ is not analysed in $t$ thought decryption or projections (i.e., $x$ does not occur inside decryption or projections in $t$), then the result of the reduction will happen regardless of $M$ and therefore the thesis holds true.
		The same happens if $\cc(M) \leq d$, since $N = pr_d(M) = M$ and therefore the lemma holds again.
		Otherwise, if $M$ is analysed in $t$ and $\cc(M) > d$ then $\cc(M) \geq \cc(N) \geq d > \mdd(t)$ and the reduction happens for $P\{M/x\}$ if and only if it happens for $P\{N/x\}$.
		From the definition of reduction for \emph{let} reduction it follows that $P'_M = Q\{\sd{t\{M/x\}}/y\}$ and $P'_N = Q\{\sd{t\{N/x\}}/y\}$.
 		\qed
\end{proof}

\begin{lemma}\label{lem:transitions}
	Let $(x)P$ be an abstraction of a finite process, $h$ be a hedge and $d = \cri(h,P)$. For every message $M$,
	\begin{equation*}
		P\{M/x\} \overset{\alpha}{\longrightarrow} P'_M\{M/x\} \iff P\{N/x\} \overset{\alpha}{\longrightarrow} P'_N\{N/x\}
	\end{equation*}
	where the same reduction rule is used,
	$N = \pi_2(\pr_d(h,M))$ and 	
	$P'_M$ depends upon the application of the substitution $\{M/x\}$ and the transition rule used.
\end{lemma}
\begin{proof} By induction on the transition rules.
	Let us start with the \emph{parallel} rule. Let $P = Q|U$. By inductive hypothesis
	\begin{equation*}
	Q\{M/x\} \overset{\alpha}{\longrightarrow} Q'\{M/x\} \iff Q\{N/x\} \overset{\alpha}{\longrightarrow} Q'\{N/x\}
	\end{equation*}
	and therefore, the application of the \emph{parallel} rule leads to
	\begin{equation*}
	(Q|U)\{M/x\} \overset{\alpha}{\longrightarrow} (Q'|U)\{M/x\} \iff (Q|U)\{N/x\} \overset{\alpha}{\longrightarrow} (Q'|U)\{N/x\}
	\end{equation*}
	and $P'_M = P'_N$. The same reasoning can be used for \emph{sum} and \emph{restriction} rules.
	For \emph{equivalence} rules, only reduction rules modifies the processes substantially. Therefore, we can use Lemma~\ref{lem:reductions} and the proof steps are equivalent to the ones of the \emph{parallel} rule.
	
	Let us consider now the \emph{input} rule, i.e., $P\{M/x\} = (t(y).Q)\{M/x\}$. Since $(t(y).Q)\{M/x\}$ must be a ground process, $\fv(t(y).Q) \subseteq \{x\}$ and $t$ can only be a name or equal to $x$. If $x \neq t \in \N$, then the transition is independent from $M$ and will also occur for $N$. Otherwise, $t = x$ and the transition will occur only if $M$ is a name. Moreover, if $M \in \N$ then $M = pr_d(M) = N$ and therefore the transition will also occur for $(t(y).Q)\{N/x\}$. It also holds that $P'_M = P'_N$.
	
	The case for \emph{output} rules is similar to \emph{input}'s, since the transition of a term $(t\langle X \rangle.Q)\{M/x\}$ only depends on the channel $t$.
	
	Let us now prove the lemma for \emph{interaction} rules. By inductive hypothesis:
	\begin{align*}
	P\{M/x\} \overset{n}{\longrightarrow} ((y)P')\{M/x\} & \iff P\{N/x\} \overset{n}{\longrightarrow} ((y)P')\{N/x\} \\
	Q\{M/x\} \overset{\overline{n}}{\longrightarrow} ((\nu m) \langle T \rangle Q')\{M/x\} & \iff
	  Q\{N/x\} \overset{\overline{n}}{\longrightarrow} ((\nu m) \langle T \rangle Q')\{N/x\}
	\end{align*}
	therefore, the application of the \emph{interaction} rule leads to
	\vspace{-1ex}
	\begin{multline*}
	(P|Q)\{M/x\} \overset{\tau}{\longrightarrow} ((\nu m)(P\{T\{M/x\}/y\} | Q'))\{M/x\}\\
	\iff
	(P|Q)\{N/x\} \overset{\tau}{\longrightarrow} ((\nu m)(P\{T\{N/x\}/y\} | Q'))\{N/x\}
	\end{multline*}
	which is our thesis, and $P'_X = ((\nu m)(P\{T\{X/x\}/y\} | Q'))$ for $X \in \{M,N\}$.
	\qed
\end{proof}

\begin{theorem}\label{th:dedged}
	Let $P$ and $Q$ be two finite processes. Then:
	\begin{equation*}
	\exists h \in \Hedges\ h \vdash P \sim Q \text{ if and only if } \exists h \in \Hedges\ h \vdash P \sim^{\cri(h,P,Q)} Q
	\end{equation*} 
\end{theorem}
\begin{proof}
	By Proposition \ref{prop:dedged}, any hedged bisimulation is also a d-hedged bisimulation; hence, it suffices to show that if there is $h \in \Hedges$ such that $h \vdash P \sim^{\cri(h,P,Q)} Q$, then there is $h \in \Hedges$ such that $h \vdash P \sim Q$. 
	This follows from Lemmata~\ref{lem:reductions} and \ref{lem:transitions}, since we already shown that reductions and transition system is the same whenever we perform a pruning of depth $\cri(h,P,Q)$.
	Moreover, it holds that
	\begin{align*}
	\R = \{ (h,P,Q) \mid & \exists h' \in \Hedges\ \exists P', Q' \in \Prt\ \exists M, N \in \M \text{ s.t.} 	\\	
		& P = P'\{M/x\}, Q = Q'\{N/y\},
	    (h', M', N') = \pr_d(h,M,N), \\
		& h' \vdash P'\{M'/x\} \sim^d Q'\{N'/y\},
		\text{for } d = \cri(h,P,Q)
		\}
	\end{align*}
	is a hedged bisimulation, since for a transition $\alpha$ it holds that
	\begin{align*}
	& P'\{M/x\} \overset{\alpha}{\longrightarrow} P''\{M/x\}\ \xLeftrightarrow{\text{Lemma~\ref{lem:transitions}}}\  P'\{M'/x\} \overset{\alpha}{\longrightarrow} P''\{M'/x\}\\ 
	&\xRightarrow{\sim^d}\ Q'\{N'/y\} \xRightarrow{\ \alpha\ } Q''\{N'/y\}\
	\xLeftrightarrow{\text{Lemma~\ref{lem:transitions}}}\
	Q'\{N/y\} \xRightarrow{\ \alpha\ } Q''\{N/y\}
	\end{align*}
	and vice versa.
	Furthermore there exists $h''$, obtained by updating $h$ with the effects of $\alpha$, such that, from Definition~\ref{def:dsim}:
	\[
	 (h''', M'', N'') = \pr_{d'}(h'', M', N') \text{ and } h''' \vdash P''\{M''/x\} \sim^{d'} Q''\{N''/y\}
	\]
	where $d' = \cri(h'', P'', Q'')$. Therefore $(h'', P''\{M/x\}, Q''\{N/y\}) \in \R$.
	\qed
\end{proof}

\begin{figure}[t]
\begin{lstlisting}[mathescape=true]
$\HB(h,P,Q)$ =
	for each $P \overset{\tau}{\longrightarrow} P'$
		select $Q \Longrightarrow Q'$ such that
			$\HB(h,P',Q') \land \HB(h^{-1},Q',P')$
	for each $P \overset{\overline{a}}{\longrightarrow} (\nu \vec{m})\langle t \rangle P'$
		select $Q \overset{\overline{b}}{\Longrightarrow} (\nu \vec{n})\langle t' \rangle Q'$ such that 
			$h_O := I(h \cup \{(a,b)\} \cup \{(\sd{t},\sd{t'})\})$ consistent and
			$\HB(h_O,P',Q')\land \HB(h_O^{-1},Q',P')$
	for each $P \overset{a}{\longrightarrow} (x)P'$
		let $d = \cri(h,P,Q)$
		select $Q \overset{b}{\Longrightarrow} (x)Q'$ and $B \subset \N$ such that
			$|B| = 2^d$, $B \cap (\fn(P) \cup \fn(Q) \cup n(h) \cup \{a,b\}) = \emptyset$,
			$h_I := h\cup\{(a,b)\}\cup id_B$ consistent and
			for each $(M,N)$ such that $\exists k \leq d:\ h_I \vdash_k M \leftrightarrow N$
				$\HB(h_I,P'\{M/x\},Q'\{N/x\}) \land \HB(h_I^{-1},Q'\{N/x\},P'\{M/x\})$
\end{lstlisting}
\caption{Algorithm for deciding hedged bisimilarity. The  \emph{select} statement implements a nondeterministic exploration of the (finite) possible choices of its argument, until the condition is satisfied; it returns $\T$ if successful, $\F$ otherwise.
}\label{fig:algorithm}
\end{figure}

Since every quantification is bounded, $d$-hedged bisimilarity is decidable on finite processes.
An algorithm is shown in Figure~\ref{fig:algorithm}.
For $P,Q$ two finite processes and $h$ a hedge (which represents the initial knowledge of the attacker, e.g. public channels, keys, etc.), $\HB(h,P,Q)\wedge\HB(h^{-1},Q,P)=\T$ if and only if there is a $d$ such that  $(h,P,Q)$ are in a $d$-hedged simulation.
Hence, by Theorem~\ref{th:dedged}:
\begin{equation*}
h \vdash P \sim Q \iff \HB(h,P,Q)\wedge \HB(h^{-1},Q,P).
\end{equation*}

\section{Conclusions and further work}\label{sec:concl}
In this paper we have proved that hedged bisimilarity is decidable on finite processes of the spi calculus.
Our algorithm, which generalizes the ideas in \cite{huttel:infinity02}, can be readily applied to different encryption/decryption schemata just by changing the congruence rules, as long as some mild conditions are satisfied.
Actually, a possible future work is to investigate the algebraic laws needed to represent in the structural congruence the properties of various encryption algorithms: often these laws are omitted from formalizations, leading to security flaws in protocols.
Another direction is to consider other fragments of the spi-calculus beyond finite processes; depth- and restriction-bounded processes are particularly promising.


\begin{thebibliography}{1}

\bibitem{AbadiG98}
Mart\'{\i}n Abadi and Andrew~D. Gordon.
\newblock A bisimulation method for cryptographic protocols.
\newblock {\em Nordic Journal of Computing}, 5(4):267--303, 1998.

\bibitem{ag:spi}
Mart{\'{\i}}n Abadi and Andrew~D. Gordon.
\newblock A calculus for cryptographic protocols: The spi calculus.
\newblock {\em Inf. Comput.}, 148(1):1--70, 1999.

\bibitem{bn:mscs05}
Johannes Borgstr{\"o}m and Uwe Nestmann.
\newblock On bisimulations for the spi calculus.
\newblock {\em Mathematical Structures in Computer Science}, 15(3):487--552,
  2005.

\bibitem{huttel:infinity02}
Hans H\"uttel.
\newblock Deciding framed bisimilarity.
\newblock In {\em Proceedings of Infinity'02}, volume~68 of {\em Electronic
  Notes in Theoretical Computer Science}, pages 1--18, 2003.

\end{thebibliography}
\end{document}